\documentclass[a4paper, 9pt, twosides]{amsart}

\usepackage{graphicx}
\usepackage{enumitem}
\usepackage{dsfont}
\usepackage[usenames, dvipsnames]{color}
\usepackage{todonotes}
\usepackage{algorithm,algorithmic}

\usepackage{newtxtext}
\usepackage[varg]{newtxmath}

\DeclareFontFamily{T1}{pzc}{}
\DeclareFontShape{T1}{pzc}{m}{it}{<-> [1.2] pzcmi8t}{}
\DeclareMathAlphabet{\mathpzc}{T1}{pzc}{m}{it}

\usepackage{multicol}

\usepackage{mathtools}
\allowdisplaybreaks
\newtheorem{theorem}{Theorem}
\newtheorem{defn}{Definition}
\newtheorem{lem}{Lemma}

\newtheorem{assump}{Assumption}

\newtheorem{remark}{Remark}

\newtheorem{prob}{Problem}

\newtheorem{experiment}{Experiment}

\newcommand{\abs}[1]{\left\lvert{#1}\right\rvert}

\newcommand{\norm}[1]{\left\lVert#1\right\rVert}
\newcommand{\pmat}[1]{\begin{pmatrix}#1\end{pmatrix}}

\newcommand{\R}{\mathbb{R}}
\newcommand{\N}{\mathbb{N}}

\renewcommand{\P}{\mathcal{P}}

\newcommand{\diag}{\text{diag}}
\newcommand{\F}{\mathcal{F}}
\newcommand{\PP}{\mathbb{P}}
\DeclareMathOperator{\minimize}{minimize}
\DeclareMathOperator{\sbjto}{subject\;to}

\renewcommand{\P}{\mathcal{P}}
\newcommand{\EE}{\mathbb{E}}
\newcommand{\Svec}{\mathcal{S}}
\newcommand{\is}{i_s}
\newcommand{\iu}{i_u}

\title{A probabilistic scheduling algorithm\\for networked control systems}
\author{Meghna Singh and Atreyee Kundu}
\thanks{MS is with the Department of Electrical and Electronics Engineering, PES University Bangalore, India, e-mail: megsingh2212@gmail.com. AK is with the Department of Electrical Engineering, Indian Institute of Technology, Kharagpur, India, Email: atreyee@ee.iitkgp.ac.in.}


\date{\today}

\begin{document}

	\begin{abstract}
	This paper deals with the design of scheduling logics for Networked Control Systems (NCSs) whose communication networks have limited capacity. We assume that only a subset of the plants can communicate with their controllers at any time instant. Our contributions are twofold. First, we present a probabilistic algorithm to design scheduling logics that, under certain conditions on the plant and controller dynamics and the capacity of the network, ensure stochastic stability of each plant in an NCS. Second, given the plant dynamics and the capacity of the network, we design static state-feedback controllers such that the conditions for stability under our scheduling logics are satisfied. The main apparatus for our analysis is a Markovian jump linear system representation of the individual plants in an NCS. Our stability conditions involve sets of matrix inequalities. We present numerical experiments to demonstrate our results.
	\end{abstract}

    \maketitle

\section{Introduction}
\label{s:intro}
     Networked Control Systems (NCSs) are an integral part of modern day Cyber-Physical Systems (CPS) and Internet of Things (IoT) applications. While these applications typically involve a large number of plants, bandwidth of shared communication networks is often limited. The scenario in which the number of plants sharing a communication network is higher than the capacity of the network is called \emph{medium access constraint}. This scenario motivates the need to allocate the communication network to each plant in a manner so that good qualitative and quantitative properties of the plants are preserved. This task of efficient allocation of a shared communication network is commonly referred to as a \emph{scheduling problem} and the corresponding allocation scheme is called a \emph{scheduling logic}. In this paper we study algorithmic design of scheduling logics for NCSs.

    The existing classes of scheduling logics can be classified broadly into two categories: \emph{static} and \emph{dynamic}. In case of the former, a finite length allocation scheme of the network is determined offline and is applied eternally in a periodic manner, while in case of the latter, the allocation of the shared network is determined based on some information about the plant (e.g., states, outputs, access status of sensors and actuators, etc.), see \cite{Walsh2001} for a detailed discussion. In this paper we consider a shift in paradigm and present probabilistic scheduling logics for NCSs.
    
   We study an NCS consisting of multiple discrete-time linear plants whose feedback loops are closed through a shared communication network. A block diagram of such an NCS is shown in Figure \ref{fig:ncs}.
    \begin{figure}[htbp]
    	\begin{center}
	\scalebox{0.6}{
	\begin{tikzpicture}[every path/.style={>=latex},base node/.style={draw,rectangle, scale = 1.4}]
	\node[base node] (a) at (-2,5) {Controller 1};
	\node[base node] (b) at (3.5,4) {Plant 1};
	\node[base node] (c) at (-2,2) {Controller 2};
	\node[base node] (d) at (3.5,1) {Plant 2};
	\node[base node] (e) at (-2,-2) {Controller N};
	\node[base node] (f) at (3.5,-3) {Plant N};	
	
	\draw (-4.5 ,5) edge (a);
	\draw (-4.5,5) edge (-4.5,4);
	\draw[->] (-4.5,4) -- (0,4);
	\draw (a) edge (0.4,5);
	\draw[->] (b) -- (5.5,4);
	\draw (5.5,4) edge (5.5,5);
	\draw[->] (1,4) -- (b);
	\draw[-.] (1,4) -- (0.4,4.5);
	\draw (5.5,5) edge (1.1,5);
	\draw[-.] (1.1,5) -- (0.6,5.5);

	\draw (-4.5 ,2) edge (c);
	\draw (-4.5,2) edge (-4.5,1);
	\draw[->] (-4.5,1) -- (0,1);
	\draw (c) edge (0.4,2);
	\draw[->] (d) -- (5.5,1);
	\draw (5.5,1) edge (5.5,2);
	\draw[->] (1,1) -- (d);
	\draw[-.] (1,1) -- (0.4,1.5);
	\draw (5.5,2) edge (1.1,2);
	\draw[-.] (1.1,2) -- (0.6,2.5);

	\draw (-4.5 ,-2) edge (e);
	\draw (-4.5,-2) edge (-4.5,-3);
	\draw[->] (-4.5,-3) -- (0,-3);
	\draw (e) edge (0.4,-2);
	\draw[->] (f) -- (5.5,-3);
	\draw (5.5,-3) edge (5.5,-2);
	\draw[->] (1,-3) -- (f);
	\draw[-.] (1,-3) -- (0.4,-2.5);
	\draw (5.5,-2) edge (1.1,-2);
	\draw[-.] (1.1,-2) -- (0.6,-1.5);

	\draw[dashed] (-0.4,-4) -- (-0.4,6);
	\draw[dashed] (2,-4) -- (2,6);
	\draw[dashed] (-0.4,-4) -- (2,-4);
	\draw[dashed] (-0.4,6) -- (2,6);

	\node (g) at (0.75,-4.5) {Communication network};
	\node (h) at (3.5,-0.5) {\(\vdots\)};

	\end{tikzpicture}
	}
	\caption{Block diagram of NCS}\label{fig:ncs}
	\end{center}
    \end{figure}
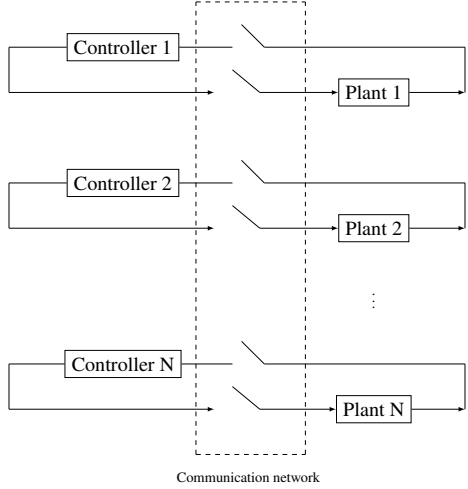
    We assume that the plants are unstable in open-loop and exponentially stable in closed-loop. Due to a limited communication capacity of the network, only a few plants can exchange information with their controllers at any instant of time. Consequently, the remaining plants operate in open-loop at every time instant. Our contributions are twofold: \\
    \begin{itemize}[label = \(\circ\),leftmargin=*]
    	\item We present an algorithm to design scheduling logics. At every instant of time, our algorithm allocates the shared network to subsets of the plants with certain probabilities. We present necessary and sufficient conditions on the plant dynamics and the capacity of the shared network under which a scheduling logic obtained from our algorithm ensures stochastic stability of each plant in the NCS.
    	\item Given plant dynamics and capacity of the shared network, we present an algorithm to design static state-feedback controllers such that the plants, their controllers and the shared network together satisfy our stability conditions.
    \end{itemize}
    The proposed stability conditions are derived using a Markovian jump linear systems modelling of the individual plants. They involve matrix inequalities and can be verified by using standard matrix inequality solver toolboxes.

    The remainder of this paper is organized as follows: In \S\ref{s:prob_stat} we formulate the problem under consideration. Our results appear in \S\ref{s:mainres}. We also describe various features of our results in this section. Numerical experiments are presented in \S\ref{s:num_ex}. We conclude in \S\ref{s:concln} with a brief discussion on future research direction.

    {\bf Notation}. \(\R\) is the set of real numbers and \(\N\) is the set of natural numbers, \(\N_0 = \N\cup\{0\}\). For two scalars \(a\) and \(b\), \(a\%b\) denotes the remainder of the operation \(a/b\). For a finite set \(C\), its cardinality is denoted by \(\abs{C}\). For a vector \(v\), \(\norm{v}\) denotes its Euclidean norm. For symmetric block matrices, \(\bigstar\) acts as ellipsis for the terms that are introduced by symmetry, \(\diag(Q_1,Q_2,\ldots,Q_n)\) denotes a block-diagonal matrix with diagonal elements \(Q_1,Q_2,\ldots,Q_n\). \(0_{d\times d}\) and \(I_{d\times d}\) denote \(d\)-dimensional \(0\)-matrix and identity matrix, respectively. We will operate in a probabilistic space \((\Omega,\F,\PP)\), where \(\Omega\) is the sample space, \(\F\) is the \(\sigma\)-algebra of events, and \(\PP\) is the probability measure.
\section{Problem statement}
\label{s:prob_stat}
    We consider an NCS with \(N\) plants whose dynamics are given by
    \begin{align}
    \label{e:plants}
        x_i(t+1) = A_i x_i(t) + B_i u_i(t),\:x_i(0) = x_i^0,\:t\in\N_0,
    \end{align}
    where \(x_i(t)\in\R^{d_i}\) and \(u_i(t)\in\R^{m_i}\) are the vectors of states and inputs of the \(i\)-th plant at time \(t\), respectively, \(i=1,2,\ldots,N\). Each plant \(i\) employs a state-feedback controller \(u_i(t) = K_i x_i(t)\), \(t\in\N_0\). The matrices \(A_i\in\R^{d_i\times d_i}\), \(B_i\in\R^{d_i\times m_i}\) and \(K_i\in\R^{m_i\times d_i}\), \(i=1,2,\ldots,N\) are constants.

     \begin{assump}
    \label{a:stability}
        The open-loop dynamics of each plant is unstable and each controller is stabilizing. More specifically, the matrices \(A_i+B_iK_i\), \(i=1,2,\ldots,N\) are Schur stable and the matrices \(A_i\), \(i=1,2,\ldots,N\) are unstable.\footnote{A matrix \(A\in\R^{d\times d}\) is Schur stable if all its eigenvalues are inside the open unit disk. We call \(A\) unstable if it is not Schur stable.}
    \end{assump}
    The controllers are remotely located and each plant communicates with its controller through a shared communication network. The network has a limited communication capacity in the sense that at any time instant, only \(M\) plants (\(0<M<N\)) can access the network. Consequently, the remaining \(N-M\) plants operate in open loop.
    \begin{assump}
    \label{a:ideal}
        The communication network is ideal in the sense that exchange of information between plants and their controllers is not affected by communication uncertainties.
    \end{assump}

    Let \(\is\) and \(\iu\) denote the stable and unstable modes of the \(i\)-th plant, respectively, \(A_{\is} = A_i+B_i K_i\) and \(A_{\iu} = A_i\), \(i=1,2,\ldots,N\). We let
    \[
        \Svec = \{s\in\{1,2,\ldots,N\}^{M}\:|\:\text{the elements of \(s\) are distinct}\}
    \]
    be the set of all subsets of \(\{1,2,\ldots,N\}\) with cardinality \(M\). We call a function \(\gamma:\N_0\to\Svec\), that specifies, at every time \(t\), \(M\) plants of the NCS which access the shared network at that time, as a \emph{scheduling logic}. Let \(r_i^0\) denote the initial mode of operation of plant \(i\), i.e., \(r_i^0 =\is\), if \(i\in\gamma(0)\) and \(r_i^0 =\iu\), if \(i\notin\gamma(0)\).
    We will focus on stochastic stability of the plants. 
    \begin{defn}
    \label{d:stability}
        The \(i\)-th plant in \eqref{e:plants} is \emph{stochastically stable} if for every initial condition \(x_i^0\in\R^{d_i}\) and initial mode of operation \(r_i^0\in\{\is,\iu\}\), we have that
            \(\displaystyle{\EE\Biggl\{\sum_{t=0}^{+\infty}\norm{x_i(t)}^{2}\:|\:x_i^0,r_i^0\Biggr\}<+\infty}\).
    \end{defn}

    Our first objective is:
    \begin{prob}
    \label{prob:main1}
        Given the matrices \(A_i\), \(B_i\), \(K_i\), \(i=1,2,\ldots,N\) and the number \(M\), design a scheduling logic, \(\gamma\), that preserves stochastic stability of each plant \(i\) in the NCS.
    \end{prob}
    Towards solving Problem \ref{prob:main1}, we will first present a probabilistic algorithm. We will then identify conditions on the matrices \(A_{\is}\), \(A_{\iu}\), \(i=1,2,\ldots,N\) and the network capacity, \(M\), such that stochastic stability of each plant \(i\) in the NCS is ensured under a scheduling logic obtained from our algorithm. 
    
    Our second objective is:
     \begin{prob}
    \label{prob:main2}
        Given the matrices \(A_i\), \(i=1,2,\ldots,N\) and the network capacity, \(M\), design static state-feedback controllers, \(K_i\), \(i=1,2,\ldots,N\), such that the conditions for stability under our scheduling logics are satisfied.
    \end{prob}
   Towards designing suitable state-feedback controllers, we will solve a set of feasibility problems involving LMIs.
\section{Main results}
\label{s:mainres}
\subsection{Stabilizing scheduling logics}
\label{ss:mainres1}
	We first present our solution to Problem \ref{prob:main1}. We will operate under the following assumption:
	\begin{assump}
	\label{a:divisibility}
	\rm{
		The total number of plants, \(N\) and the capacity of the shared communication network, \(M\) together satisfy \(N\%M = 0\).
	}
	\end{assump}
	
	Assumption \ref{a:divisibility} ensures that the total number of plants, \(N\), in the NCS is divisible by the capacity of the shared network, \(M\). In other words, the \(N\) plants can be divided into an integer number of chunks of \(M\) plants. Let \(v = N/M\). Towards designing a scheduling logic, we rely on disjoint sets \(c_1,c_2,\ldots,c_v\in\Svec\) and scalars \(p_{c_1},p_{c_2},\ldots\),\(p_{c_{v}}\in]0,1[\) that satisfy \(\displaystyle{\sum_{j=1}^{v}p_{c_{j}}}=1\). 
	 
	 Suppose that \(c_1,c_2,\ldots,c_v\) and \(p_{c_1},p_{c_2},\ldots\),\(p_{c_{v}}\) are fixed. A scheduling logic, \(\gamma\), is generated as follows: at each time instant \(t=0,1,2,\ldots\), we allocate the shared network to the plants in \(c_j\) with probability \(p_{c_j}\), \(j\in\{1,2,\ldots,v\}\). This procedure is summarized in Algorithm \ref{algo:sched_design}.\footnote{We will discuss how to choose the quantities \(c_1,c_2,\ldots,c_v\) and \(p_{c_1},p_{c_2},\ldots\),\(p_{c_{v}}\) favourably in a moment.}
	 \begin{algorithm}
    \begin{algorithmic}[1]
        \STATE Set \(v = N/M\).
        \STATE Pick \(c_1,c_2,\ldots,c_v\in\Svec\) such that \(c_j\cap c_k=\emptyset\) for all \(j,k=1,2,\ldots,v\), \(j\neq k\).
        \STATE Pick \(p_{c_1},p_{c_{2}},\ldots,p_{c_{v}}\in]0,1[\) such that \(\displaystyle{\sum_{j=1}^{v}p_{c_{j}}} = 1\).
        \FOR {\(t=0,1,2,\ldots\)}
            \STATE Set \(\gamma(t) = c_j\) with probability \(p_{c_{j}}\).
        \ENDFOR
    \caption{Design of a scheduling logic}\label{algo:sched_design}
    \end{algorithmic}
    \end{algorithm}
The following theorem provides necessary and sufficient conditions on the matrices, \(A_i\), \(B_i\), \(K_i\), \(i=1,2,\ldots,N\) and the capacity of the network, \(M\), under which a scheduling logic, \(\gamma\), obtained from Algorithm \ref{algo:sched_design} ensures stochastic stability of each plant in the NCS.
    \begin{theorem}
    \label{t:mainres1}
        Consider an NCS described in \S\ref{s:prob_stat}. Suppose that Assumption \ref{a:divisibility} holds. Let \(\gamma\) be a scheduling logic obtained from Algorithm \ref{algo:sched_design}. Each plant \(i\) in \eqref{e:plants} is stochastically stable under \(\gamma\) if and only if the following conditions hold:
            for each \(i\in c_j\), \(j=1,2,\ldots,v\), there exist symmetric and positive definite matrices \(P_k\in\R^{d_i\times d_i}\), \(k=\is,\iu\), such that
                \begin{align}
                \label{e:maincondn}
                    A_k^\top \P^i A_k - P_k \prec 0,
                \end{align}
                where \(\P^i = p_{c_j}P_{\is} + (1-p_{c_{j}})P_{\iu}\).
     \end{theorem}

	\begin{remark}
	\label{rem:expln1}
    Condition \eqref{e:maincondn} involves properties of the matrices \(A_i\), \(B_i\) and \(K_i\), \(i=1,2,\ldots,N\), the disjoint sets \(c_j\), \(j\in\{1,2,\ldots,v\}\) and the probabilities \(p_{c_{j}}\), \(j\in\{1,2,\ldots,v\}\). It relies on the existence of symmetric and positive definite matrices, \(P_k\), \(k=\is,\iu\), \(i=1,2,\ldots,N\) that together with the matrices \(A_i\), \(B_i\), \(K_i\), \(i=1,2,\ldots,N\) and the probabilities \(p_{c_j}\), \(j=1,2,\ldots,v\) corresponding to the subset \(c_j\), \(j\in\{1,2,\ldots,v\}\) that plant \(i\) appears in, satisfy a set of matrix inequalities. Notice that with the quantities \(A_i\), \(B_i\), \(K_i\), \(c_j\), \(p_{c_j}\), \(j=1,2,\ldots,v\) known, the set of inequalities in \eqref{e:maincondn} can be solved by employing standard Linear Matrix Inequalities solvers.
    \end{remark}
    
    \begin{remark}
	\label{rem:expln2}
    Fix a scheduling logic, \(\gamma\), obtained from Algorithm \ref{algo:sched_design}. Theorem \ref{t:mainres1} is necessary and sufficient in the following sense: if condition \eqref{e:maincondn} holds, then \(\gamma\) is stabilizing, and if \(\gamma\) is stabilizing, then condition \eqref{e:maincondn} holds.  
    \end{remark}
    
    Towards proving Theorem \ref{t:mainres1}, we will utilize the following auxiliary result:
    \begin{lem}
    \label{lem:auxres}
        Suppose that Assumption \ref{a:divisibility} holds. Then the following are true:
        \begin{enumerate}[label = \roman*), leftmargin = *]
            \item \(\displaystyle{\bigcup_{j=1}^{v}c_j} = \{1,2,\ldots,N\}\), and
            \item for each \(i\in\{1,2,\ldots,N\}\), there exists exactly one \(j\in\{1,2,\ldots,v\}\) such that \(i\in c_j\).
        \end{enumerate}
    \end{lem}
     \begin{proof}
        i) Assume, by contradiction, that
        \(
            \displaystyle{\bigcup_{j=1}^{v}c_j} \neq \{1,2,\ldots\),\(N\}.
        \)
        By construction of \(c_j\), \(j=1,2,\ldots,v\), it must then be true that
            \(\displaystyle{\bigcup_{j=1}^{c}c_j\subset\{1,2,\ldots,N\}}\).
        We have \(\abs{c_j} = M\) for each \(j=1,2,\ldots,v\). Thus, \(\displaystyle{\abs{\bigcup_{j=1}^{v}c_j}} = vM=(N/M)M\). Since \(N\% M = 0\), we have \(\displaystyle{\abs{\bigcup_{j=1}^{v}c_j}} = N\). Then it must hold that there exist \(\ell\in\{1,2,\ldots,N\}\) and \(j_1,j_2,\ldots,j_q\in\{1,2,\ldots,v\}\) such that \(\ell\in c_{j_{m}}\) for each \(m=1,2,\ldots,q\). But this contradicts the fact that \(c_j\), \(j=1,2,\ldots,v\) are disjoint sets. Consequently, it must be true that \(\displaystyle{\bigcup_{j=1}^{v}c_j} = \{1,2,\ldots,N\}\).

        ii) Since \(v = N/M\) and the sets \(c_j\), \(j=1,2,\ldots,v\), are disjoint, the assertion follows at once.
    \end{proof}

    \begin{proof}{(of Theorem \ref{t:mainres1})}: 
    Fix a scheduling logic, \(\gamma\), obtained from Algorithm \ref{algo:sched_design}. We will show that condition \eqref{e:maincondn} is necessary and sufficient for stability of each plant \(i\) in \eqref{e:plants} under \(\gamma\).

    Fix \(j\in\{1,2,\ldots,v\}\) and \(i\in c_j\). 
    By Lemma \ref{lem:auxres} ii), \(i\) appears in exactly one \(c_j\). We model the plant \(i\) under \(\gamma\) as follows:
    \begin{align}
    \label{e:i-swsys}
        x_i(t+1) = A_{\sigma_i(t)}x_i(t),\:\sigma_i(t)\in\{\is,\iu\}.
    \end{align}
    Notice that \eqref{e:i-swsys} is a Markovian jump linear system whose set of subsystems is \(\{\is,\iu\}\) and the transition function \(\sigma_i\in\N_0\to\{\is,\iu\}\) satisfies \(\sigma_i(t) =\is\), if \(i\in\gamma(t)\) and \(\sigma_i(t) =\iu\), if \(i\notin\gamma(t)\).
    In particular, \(\sigma_i\) is a Markov chain, defined on \((\Omega,\F,\PP)\), taking values in \(\{\is,\iu\}\) with transition probability matrix
    \(
        \Pi_i = \pmat{\pi_{\is\is} & \pi_{\is\iu}\\\pi_{\iu\is} & \pi_{\iu\iu}},
    \)
    where
    \begin{align*}
        \begin{aligned}
        \pi_{\is\is} &= \PP(\sigma_i(t+1) = \is\:|\:\sigma_i(t) = \is) = p_{c_j},\\
        \pi_{\is\iu} &= \PP(\sigma_i(t+1) = \iu\:|\:\sigma_i(t) = \is) = 1-p_{c_j},\\
        \pi_{\iu\is} &= \PP(\sigma_i(t+1) = \is\:|\:\sigma_i(t) = \iu) = p_{c_j},\\
        \pi_{\iu\iu} &= \PP(\sigma_i(t+1) = \iu\:|\:\sigma_i(t) = \iu) = 1-p_{c_j},
        \end{aligned}
        \:\:i\in c_j.
    \end{align*}

    By \cite[Lemma 2]{Zhang2009}, the switched system \eqref{e:i-swsys} is stochastically stable if and only if the following conditions hold:
    \begin{align}
        \label{e:pf1_step1}A_{\is}^\top \bigl(\pi_{\is\is}P_{\is}+\pi_{\is\iu}P_{\iu}\bigr)A_{\is} - P_{\is} \prec 0,\\
        \intertext{and}
        \label{e:pf1_step2}A_{\iu}^\top \bigl(\pi_{\iu\is}P_{\is}+\pi_{\iu\iu}P_{\iu}\bigr)A_{\iu} - P_{\iu} \prec 0,
    \end{align}
    where \(P_{\is}\), \(P_{\iu}\in\R^{d_i\times d_i}\) are symmetric and positive definite matrices. We have
        \(\pi_{\is\is}P_{\is}+\pi_{\is\iu}P_{\iu} = p_{c_j}P_{\is} + (1-p_{c_{j}})P_{\iu}
        =\pi_{\iu\is}P_{\is}+\pi_{\iu\iu}P_{\iu} = \P^i\).
    Clearly, stochastic stability of plant \(i\) is equivalent to the conditions \eqref{e:pf1_step1}-\eqref{e:pf1_step2}.

    Since \(i\in c_j\) and \(j\in\{1,2,\ldots,v\}\) were chosen arbitrarily, stochastic stability of each plant \(\displaystyle{i\in\bigcup_{j=1}^{v}c_j}\) under \(\gamma\) is immediate. In view of Lemma \ref{lem:auxres} i), this completes our proof of Theorem \ref{t:mainres1}.
    \end{proof}

     \begin{remark}
    \label{rem:compa1}
    \rm{
     	Recall that a Markovian jump linear system is a switched system \cite[Section 1.1.2]{Liberzon} with linear subsystems; its switching logic is stochastic and can be described by a Markov chain. Switched systems  with both deterministic and stochastic switching logics have been employed to design scheduling algorithms for NCSs with communication limitations and uncertainties earlier in the literature, see e.g., \cite[Remark 11]{quevedo2020} for a detailed discussion. The primary difference of our work with the existing literature is that our scheduling algorithm is probabilistic while most of the existing design of scheduling logics by employing switched systems modelling of plants relies on purely deterministic techniques. In case of the latter, stochastic behaviour of the switching logics arises from probabilistic assumptions on the communication uncertainties typically leading to non-homogeneous Markov chains, see e.g., \cite{ghi} where a probabilistic data loss model is considered. In the current work stochastic behaviour of the switching logics arises from probabilistic scheduling logic and the switching logics are time homogeneous Markov chains.
        }
    \end{remark}
    
      \begin{remark}
    \label{rem:compa2}
    \rm{
        Qualitative and quantitative properties of conti-nuous-time linear plants communicating with their controllers under a \emph{pre-specified} stochastic scheduling logic have been studied in \cite{Fridman2015}. The problem considered in this paper differs from the said setting due to the following two reasons: (a) we focus on \emph{designing} stabilizing probabilistic scheduling logics, and (b) our plant dynamics evolve in discrete-time.
        }
    \end{remark}

    \begin{remark}
    \label{rem:compa3}
    \rm{
        Notice that our design of scheduling logics is neither static nor dynamic (a description of these terms are given in Section \ref{s:intro}). Indeed, we neither repeat a finite length allocation scheme nor take properties of the plants or other components in the NCS into consideration at every time instant. This is not surprising as the proposed design technique is solely probabilistic.
        }
    \end{remark}
    
     For selecting disjoint sets \(c_j\), \(j\in\{1,2,\ldots,v\}\) and the probabilities \(p_{c_{j}}\), \(j\in\{1,2,\ldots,v\}\) such that condition \eqref{e:maincondn} holds, we employ an exhaustive search over all combinations of \(v\)-many disjoint sets \(\overline{c}_j\in\Svec\), \(j=1,2,\ldots,v\) and probabilities \(\overline{p}_{\overline{c}_j}\in]0,1[\), \(j=1,2,\ldots,v\) such that \(\displaystyle{\sum_{j=1}^{v}\overline{p}_{\overline{c}_{j}}} = 1\) holds.\footnote{A search over all combinations of \(v\)-many disjoint sets \(\overline{c}_{j}\in\Svec\), \(j=1,2,\ldots,v\) suffices in view of Lemma \ref{lem:auxres}.} The interval \(]0,1[\) is sampled with a (small enough) step size \(h > 0\). Let \(r\) be the biggest integer satisfying \(rh < 1\). For all combinations of \(v\)-many disjoint sets \(\overline{c}_j\in\Svec\), \(j=1,2,\ldots,v\) and all choices of \(\overline{p}_{\overline{c}_j}\in\{h,2h,3h,\ldots,rh\}\) such that \(\displaystyle{\sum_{j=1}^{v}\overline{p}_{\overline{c}_{j}}} = 1\), we solve a feasibility problem for all plants \(i=1,2,\ldots,N\). It outputs, if exist, symmetric and positive definite matrices, \(P_k\), \(k=\is,\iu\), \(i=1,2,\ldots,N\) that together with the matrices \(A_i\), \(B_i\), \(K_i\), \(i=1,2,\ldots,N\) and the probabilities \(\overline{p}_{\overline{c}_j}\), \(j=1,2,\ldots,v\) satisfy condition \eqref{e:maincondn}. If an output is obtained, then we assign \(c_j = \overline{c}_j\) and \(p_{c_j} = \overline{p}_{\overline{c}_j}\), \(j=1,2,\ldots,v\). Otherwise, we do not have suitable inputs for Algorithm \ref{algo:sched_design}. The procedure is summarized in Algorithm \ref{algo:exhaust_search}.
    	 \begin{algorithm}
    \begin{algorithmic}[1]
        \STATE Construct the set \(\Svec\).
        \STATE Fix a step size \(h > 0\) (small enough). Compute \(r\) as the biggest integer satisfying \(rh < 1\).
        \FOR {all \(\overline{c}_1,\overline{c}_2,\ldots,\overline{c}_v\in\Svec\) such that \(\overline{c}_j\cap \overline{c}_k=\emptyset\) for all \(j,k=1,2,\ldots,v\), \(j\neq k\)}
        		\FOR {\(\overline{p}_{\overline{c}_1} = h,2h,\ldots,rh\)}
			\FOR {\(\overline{p}_{\overline{c}_2} = h,2h,\ldots,rh\)}
				\STATE \(\vdots\)
				\FOR {\(\overline{p}_{c_v} = h,2h,\ldots,rh\)}
					\IF {\(\displaystyle{\sum_{j=1}^{v}\overline{p}_{\overline{c}_{j}}} = 1\)}
						\STATE Solve the following feasibility problem for \(P_k\), \(k=\is,\iu\), \(i=1,2,\ldots,N\):
						\small
						\begin{align}
						\label{e:feas_prob}
							\hspace*{-3cm}\minimize\:\:&\:\:1\\
							\sbjto\:\:&\:\:
							\begin{cases}
								A_{\is}^\top \biggl(\overline{p}_{\overline{c}_j}P_{\is}+\bigl(1-\overline{p}_{\overline{c}_j}\bigr)P_{\iu}\biggr) A_{\is} \nonumber\\- P_{\is} \prec 0,\nonumber\\
								A_{\iu}^\top \biggl(\overline{p}_{\overline{c}_j}P_{\is}+\bigl(1-\overline{p}_{\overline{c}_j}\bigr)P_{\iu}\biggr) A_{\iu}\nonumber\\ - P_{\iu} \prec 0,\nonumber\\
								P_{\is} = P_{\is}^\top,\:\:P_{\is}\succ 0,\nonumber\\
								P_{\iu} = P_{\iu}^\top,\:\:P_{\iu}\succ 0,\nonumber\\
								\kappa I_{d_i\times d_i}\preceq P_{\is}, P_{\iu} \preceq I_{d_i\times d_i},\\
								\kappa > 0\:(\text{small}),\\
								i=1,2,\ldots,N.
							\end{cases}
						\end{align}
						\normalsize
						\STATE If a solution to \eqref{e:feas_prob} is obtained, then go to Step \ref{step:final}.
					\ENDIF
				\ENDFOR
			\ENDFOR
		\ENDFOR
		  \ENDFOR
		\STATE \label{step:final} Set \(c_j = \overline{c}_j\) and \(p_{c_j} = \overline{p}_{\overline{c}_j}\), \(j=1,2,\ldots,v\) and exit.
       \caption{Selection of \(c_j\) and \(p_{c_{j}}\), \(j = 1,2,\ldots,v\)}\label{algo:exhaust_search}
    \end{algorithmic}
    \end{algorithm}
    
    \begin{remark}
    \label{rem:bounds_search}
    \rm{
    	Notice that the conditions \(\kappa I_{d_i\times d_i}\preceq P_{\is}, P_{\iu} \preceq I_{d_i\times d_i}\) in the feasibility problem \eqref{e:feas_prob} is not inherent to the set of inequalities \eqref{e:maincondn}. It is included for numerical reasons. In particular, \(\kappa I_{d_i\times d_i}\preceq P_{\is}, P_{\iu}\) limits the condition numbers of \(P_{\is}\) and \(P_{\iu}\) to \(\kappa^{-1}\), and the condition \(P_{\is}, P_{\iu} \preceq I_{d_i\times d_i}\) guarantees that the set of feasible \(P_{\is}\), \(P_{\iu}\) is bounded. Here, we have \(i=1,2,\ldots,N\).
    }
    \end{remark}
    
    \begin{remark}
    \label{rem:complexity}
    \rm{
    	Algorithm \ref{algo:exhaust_search} has a large computational complexity when the number of plants, \(N\) and their dimensions, \(d_i\), \(i=1,2,\ldots,N\) are large. However, selection of \(c_j\) and \(p_{c_j}\), \(j=1,2,\ldots,v\) is an offline process. Indeed, they are to be chosen only once prior to the generation of a scheduling logic.
    }
    \end{remark}
    
    Suppose that suitable sets \(c_j\) and scalars \(p_{c_{j}}\), \(j=1,2,\ldots,v\) are obtained from Algorithm \ref{algo:exhaust_search}. A next natural question is: how do we choose an element \(c_j\) with a probability \(p_{c_j}\), \(j\in\{1,2,\ldots,v\}\) at every instant of time \(t\in\N_0\)? Clearly, using a standard random number \(j\in\{1,2,\ldots,v\}\) generator is not sufficient as we have a probability \(p_{c_j}\) associated to every \(c_j\). We employ Algorithm \ref{algo:impl} for this purpose.
    \begin{algorithm}
    \begin{algorithmic}[1]
        \STATE Fix \(T\in\N\) (large enough).
        \FOR {\(j=1,2,\ldots,v\)}
            \STATE Set the frequency of occurrence of \(c_j\) as \(f_{c_j} = p_{c_j}\times T\).
        \ENDFOR
        \STATE Construct a set \(TEMP\) that contains \(f_{c_{j}}\) instances of \(c_j\), \(j=1,2,\ldots,v\), i.e.,
           \small \(\displaystyle{TEMP = \bigcup_{j=1}^{v}\biggl\{c_j^{1},c_j^{2},\ldots,c_j^{f_{c_j}}\biggr\}}\).\normalsize
        \FOR {\(t = 0,1,\ldots,T-1\)}
            \STATE Pick an element \(r\) from \(TEMP\) uniformly at random, set \(\gamma(t) = r\) and \(TEMP = TEMP\setminus\{r\}\).
        \ENDFOR
    \caption{Implementation of Algorithm \ref{algo:sched_design}}\label{algo:impl}
    \end{algorithmic}
    \end{algorithm}
    It involves four steps: First, a time horizon \(\{0,1,\ldots,T-1\}\) is fixed, where \(T\in\N\) is a large number. Second, the frequency of occurrence of each \(c_j\) in \(\{0,1,\ldots,T-1\}\) is computed as \(f_{c_j} = p_{c_j}\times T\), \(j=1,2,\ldots,v\). Notice that
       \(\displaystyle{\sum_{j=1}^{v}f_{c_j} = \sum_{j=1}^{v}p_{c_j}\times T = T\sum_{j=1}^{v}p_{c_{j}} =}\)\({T}\).
    Third, a set \(TEMP\) is created with \(f_{c_j}\)-many instances of \(c_j\), \(j=1,2,\ldots,v\). It follows that \(\abs{TEMP} = T\). Fourth, at each time \(t=0,1,\ldots,T-1\), an element \(r\) from \(TEMP\) is chosen uniformly at random, is assigned to \(\gamma(t)\), and the set \(TEMP\) is updated to be \(TEMP\setminus\{r\}\). Clearly, the sequence \(\gamma(0),\gamma(1),\ldots\), \(\gamma(T-1)\) obeys the frequency of occurrence, \(f_{c_j}\), for the set \(c_j\), \(j=1,2,\ldots,v\). Our procedure for implementing Algorithm \ref{algo:sched_design}, however, has a large memory requirement when the numbers \(v\) and \(T\) are very large.
    
    \begin{remark}
    \label{rem:inhomo1}
    \rm{
    	Recall that we have been operating under Assumption \ref{a:divisibility}. The requirement for this assumption is purely technical and specific to our key apparatus of analysis. With \(N\%M\neq 0\) and the probabilistic logic for the selection of plants employed in Algorithm \ref{algo:sched_design}, the individual plants cannot be modelled as Markovian jump linear systems whose transition process is a time homogeneous Markov chain. Indeed, consider the Markovian jump linear system modelling of each plant in NCS under a scheduling logic, \(\gamma\), obtained from Algorithm \ref{algo:sched_design} as employed in our proof of Theorem \ref{t:mainres1}. We could use a time homogeneous Markov chain under the assertion of Lemma \ref{lem:auxres} ii). If \(N\%M\neq 0\) and \(v=\lceil N/M\rceil\), then there exists at least one \(i\in\{1,2,\ldots,N\}\) that appears in more than one \(c_j\), \(j=1,2,\ldots,v\). Consequently, the probability of transition to mode \(\is\) are possibly multiple for different \(j\in\{1,2,\ldots,v\}\) such that \(i\) appears in \(c_j\). As a result, the transition probability matrix, \(\Pi_i\), is no longer constant. A time inhomogeneous Markov chain with a time-varying transition probability matrix is suitable for the setting where no restriction on how the numbers \(N\) and \(M\) are connected is imposed. This general case is beyond the scope of this paper and we identify it as a topic for future work.
    }
    \end{remark}
    
    \begin{remark}
    \label{rem:(dis)adv}
    \rm{
    	In this paper we have proposed a probabilistic algorithm for scheduling NCSs. Our tool is new and differs from the techniques existing currently in the literature. We highlight the following features:
	\begin{enumerate}[label = (\alph*), leftmargin=*]
		\item In terms of offline computations required prior to the implementation of the scheduling logic, our technique is close to a static scheduling mechanism. Indeed, in case of the latter, a finite length allocation scheme is computed offline and is repeated eternally, while we compute a finite set of disjoint sets and probabilities for their activation and use the quantities eternally. 
		\item Unlike a dynamic scheduling mechanism, our technique does not consider properties of the plants and/or the communication network and/or other components in the NCS at every instant of time. 
		\item Our technique does not adapt to unforeseen/sudden faults in the system. The mechanism needs to be interrupted externally, and a new set of disjoint sets and their associated probabilities are to be fed.
	\end{enumerate}
    }
    \end{remark}
    
    Notice that the matrices \(A_i\), \(B_i\), \(i=1,2,\ldots,N\) and the capacity of the network, \(M\) are beyond our control, whereas there is an element of choice associated to the matrices \(K_i\), \(i=1,2,\ldots,N\), the sets \(c_j\), \(j=1,2,\ldots,N\) and the probabilities \(p_{c_{j}}\), \(j=1,2,\ldots,v\). We address this matter in our solution to Problem \ref{prob:main2}. 
    
    We now present Algorithm \ref{algo:controller_design} to design state-feedback controllers, \(K_i\), \(i=1,2,\ldots,N\), for the plants in the NCS. This is our solution to Problem \ref{prob:main2}. The algorithm first employs the matrices \(A_i\), \(i=1,2,\ldots,N\), the chosen disjoint sets \(c_j\), \(j=1,2,\ldots,v\) and their corresponding probabilities of allocation of the shared network, \(p_{c_j}\), \(j=1,2,\ldots,v\), to obtain symmetric and positive definite matrices, \(P_{\is}\), \(P_{\iu}\), \(i=1,2,\ldots,N\) that satisfy condition \eqref{e:maincondn} with \(k=\iu\), \(i=1,2,\ldots,N\). It then utilizes the matrices \(A_i\), \(B_i\), \(P_{\is}\), \(P_{\iu}\), \(i=1,2,\ldots,N\), to arrive at suitable controllers, \(K_i\), \(i=1,2,\ldots,N\) such that condition \eqref{e:maincondn} holds with \(k=\is\), \(i=1,2,\ldots,N\). A set of feasibility problems is employed for this design. The matrix inequalities involved in Algorithm \ref{algo:controller_design} can be solved by employing standard linear matrix inequalities and bilinear matrix inequalities toolboxes. The following theorem asserts that state-feedback controllers obtained from Algorithm \ref{algo:controller_design} meet our requirements.

	  \begin{algorithm}
    \begin{algorithmic}[1]
       \FOR {\(i=1,2,\ldots,N\)}
            \STATE Solve the following feasibility problem for \(P_{\is}\), \(P_{\iu}\in\R^{d_i\times d_i}\):
            \small
            \begin{align}
            \label{e:feasprob1}
                \minimize\:\:&\:\:1\\
                \sbjto\:\:&\:\:
                \begin{cases}
                    A_{\iu}^\top \P^iA_{\iu} - P_{\iu} \prec 0,\nonumber\\
                    P_{\is} = P_{\is}^\top, P_{\iu} = P_{\iu}^\top,\nonumber\\
                    P_{\is},P_{\iu} \succ 0,\nonumber\\
                    \kappa I_{d_i\times d_i}\preceq P_{\is}, P_{\iu} \preceq I_{d_i\times d_i},\:\kappa > 0\:(\text{small}).
                \end{cases}
            \end{align}
            \normalsize
            \IF {the feasibility problem \eqref{e:feasprob1} admits a solution}
                \STATE Solve the following feasibility problem for \(Y_i\in\R^{m_i\times d_i}\):
                \small
                \begin{align}
            \label{e:feasprob2}
                \hspace*{-3cm}\minimize\:\:&\:\:1\\
                \sbjto\:\:&\:\:
                \begin{cases}
                    &\bigl(A_{i}P_{\is}^{-1}+B_i Y_i\bigr)^\top(\P^i)^{-1}\bigl(A_{i}P_{\is}^{-1}+B_i Y_i\bigr)\nonumber\\
                    &- P_{\is}^{-1} \prec 0.\nonumber
                    \end{cases}
            \end{align}
            \normalsize
            \IF {the feasibility problem \eqref{e:feasprob2} admits a solution}
                \STATE Compute \(K_i\) as follows:
                \small
                \begin{align}
                \label{e:controller_comp}
                    K_i = Y_i P_{\is}.
                \end{align}
                \normalsize
            \ENDIF
            \ENDIF
       \ENDFOR
    \caption{Design of static state-feedback controllers}\label{algo:controller_design}
    \end{algorithmic}
    \end{algorithm}

    \begin{theorem}
    \label{t:mainres2}
        Consider an NCS described in \S\ref{s:prob_stat}. Suppose that Assumption \ref{a:divisibility} holds. Let the matrices \(A_i\), \(B_i\), \(i=1,2,\ldots,N\), the sets \(c_j\), \(j=1,2,\ldots,v\) and the probabilities \(p_{c_j}\), \(j=1,2,\ldots,v\), be given. Suppose that the state-feedback controllers, \(K_i\), \(i=1,2,\ldots,N\) are computed as \eqref{e:controller_comp}. Then condition \eqref{e:maincondn} holds.
    \end{theorem}
    
     \begin{proof} 
     Fix \(j\in\{1,2,\ldots,v\}\) and \(i\in c_j\).
    Suppose that there exists a solution \(P_{\is}\), \(P_{\iu}\) to the feasibility problem \eqref{e:feasprob1}. By Schur complement, the inequality
    \begin{align}
    \label{e:pf2_step1}
        A_{\iu}^\top \P^i A_{\iu} - P_{\iu} \prec 0
    \end{align}
    is equivalent to
  \small
        \(\pmat{-\P^i & \P^i A_{\iu}\\\bigstar & -P_{\iu}} \prec 0\).
    \normalsize
    We need to design \(K_i\) such that the following inequality holds:
    \begin{align}
    \label{e:pf2_step3}
        \pmat{-\P^i & \P^i A_{\is}\\\bigstar & -P_{\is}} \prec 0.
    \end{align}
    Let \(K_i\) be computed as described in \eqref{e:controller_comp}. We perform a congruence transform to the left-hand side of \eqref{e:pf2_step3} by \small\(\diag\pmat{(\P^i)^{-1},P_{\is}^{-1}}\)\normalsize and obtain
       \( \pmat{-(\P^i)^{-1} & (A_i P_{\is}^{-1}+B_i Y_i)\\\bigstar & -P_{\is}^{-1}}\).
    From \eqref{e:pf2_step3} it follows that the above quantity is negative definite. By Schur complement, we have
    \begin{align}
    \label{e:pf2_step4}
        (A_i P_{\is}^{-1}+B_i Y_i)^\top (\P^i)^{-1}(A_i P_{\is}^{-1}+B_i Y_i) - P_{\is}^{-1} \prec 0.
    \end{align}
    Consequently, if the feasibility problem \eqref{e:feasprob2} admits a solution \(Y_i\), then \(K_i\) computed as \eqref{e:controller_comp} satisfies \eqref{e:pf2_step3}. Conditions \eqref{e:pf2_step1} and \eqref{e:pf2_step4} together lead to \eqref{e:maincondn}.

    Since \(i\in c_j\) was chosen arbitrarily, it follows that condition \eqref{e:maincondn} holds for each plant \(i\in c_j\). Moreover, since \(j\in\{1,2,\ldots,v\}\) was chosen arbitrarily, we have that condition \eqref{e:maincondn} holds for each plant \(\displaystyle{i\in\bigcup_{j=1}^{v}c_j}\). By Lemma \ref{lem:auxres} i), the assertion of Theorem \ref{t:mainres2} follows.
    \end{proof}

      \begin{remark}
    \label{rem:compa4}
    \rm{
        Our design of static state-feedback controllers in Algorithm \ref{algo:controller_design} involves standard linear algebraic techniques for stabilization of Markovian jump linear systems. The analysis is similar in spirit to \cite{Zhang2009}, where stability of \emph{a} Markovian jump linear system was considered. We deal with a more general setting of \emph{simultaneous} design of state-feedback controllers for \(N\) such systems.
        }
    \end{remark}

%
    We now present numerical experiments to demonstrate our results.
\section{Numerical experiments}
\label{s:num_ex}
	Our first experiment involves linearized models of benchmark control systems.
	\begin{experiment}
	\label{ex:num_exp1}
	\rm{
    Consider an NCS with number of plants, \(N = 2\) and capacity of the shared communication network, \(M = 1\).
    \begin{itemize}[label = \(\circ\), leftmargin = *]
        \item Plant \(i=1\) is a discretized version of a linearized batch reactor system presented in \cite[\S IVA]{Walsh2002} with sampling time \(0.05\) units of time.  We have
           \begin{align*}
        A_1 &= \pmat{1.0795 & -0.0045 & 0.2896 & -0.2367\\-0.0272 & 0.8101 & -0.0032 & 0.0323\\
        0.0447 & 0.1886 & 0.7317 & 0.2354\\0.0010 & 0.1888 & 0.0545 & 0.9115},
 	B_1 = \pmat{0.0006 & -0.0239\\0.2567 & 0.0002\\0.0837 & -0.1346\\0.0837 & -0.0046},\\
 	\end{align*}
        \item Plant \(i=2\) is a discretized version of a linearized inverted pendulum system presented in \cite[\S 4]{Rehbinder2004} with sampling time \(0.05\) units of time. We have
      \begin{align*}
  		A_2 &= \pmat{1.0123 & 0.0502\\0.4920 & 1.0123},\:B_2 = \pmat{0.0123\\0.4920}.
    \end{align*}
    \end{itemize}
	
	Notice that the plants are open-loop unstable and \(N\%M = 0\). We compute \(v = N/M = 2\). Let \(c_1 = \{1\}\), \(c_2 = \{2\}\) and \(p_{c_1} = p_{c_2} = 0.5\). We first design static state-feedback controllers, \(K_i\), \(i=1,2\) such that condition \eqref{e:maincondn} holds for each \(i\) in \eqref{e:plants}. We employ Algorithm \ref{algo:controller_design} for this purpose. We obtain
    \begin{align*}
        K_1 = \pmat{0.0152761 & -0.8159748 & -0.2394377 & -0.7514747\\
   2.3245781 &  0.0798596  & 1.622477 &  -1.0654847}
    \end{align*}
    and
    \begin{align*}
        K_2 = \pmat{-2.3973087 & -1.4308615}.
    \end{align*}
    We have
    \begin{align*}
        P_{1s} &= \pmat{974.82022  & 115.25221 &  693.51383 & -223.88521\\
   115.25221 &  1022.0729  & 160.38138 &  109.95335\\
   693.51383 &  160.38138  & 768.15463 & -219.94088\\
  -223.88521 &  109.95335  & -219.94088 &  1250.1576},\\
        P_{1u} &= \pmat{1678.8234 &  300.05968 &  1271.4766 & -378.75625\\
   300.05968 &  1465.4904 &  391.07683  & 368.29291\\
   1271.4766 &  391.07683 &  1213.8238 & -279.44358\\
  -378.75625 &  368.29291 & -279.44358 &  1483.7789},\\
        Y_{1} &= \pmat{ 0.0005645 & -0.0006647 & -0.0008519 & -0.0005914\\
   0.0024236 & -0.0001203 & -0.0001764 & -0.0004387},
    \end{align*}
    and
    \begin{align*}
        P_{2s} &= \pmat{1717.7113 &  138.39564\\138.39564 &  50.218134},\\
        P_{2u} &= \pmat{2580.3612 & 512.67656\\512.67656 &  184.31981},\\
        Y_{2} &= \pmat{0.0011569 & -0.0316812}.
    \end{align*}
    It follows that
    \begin{align*}
        A_{1_s}^\top \P^1 A_{1_s} - P_{1_s}
        &=\pmat{-51.553004 & -7.8596573 & -69.500984 & -13.199701\\
  -7.8596573 & -480.12758 & -40.530729 &  37.248709\\
  -69.500984 & -40.530729 & -230.28674 &  106.35376\\
  -13.199701 &  37.248709 &  106.35376 & -300.5394}
    \prec 0_{d_1\times d_1},
	\end{align*}
    \begin{align*}
        A_{1_u}^\top \P^1 A_{1_u} - P_{1_u}
        &=\pmat{-48.482389 &  0.1252562 & -62.165288 & -15.778984\\
   0.1252562 & -428.29213 & -25.838462 &  53.124646\\
  -62.165288 & -25.838462 & -182.71532 &  86.522247\\
  -15.778984 &  53.124646 &  86.522247 & -289.02844}
  \prec 0_{d_1\times d_1},
    \end{align*}
    and
   \begin{align*}
        A_{2_s}^\top \P^2 A_{2_s} - P_{2_s}
        &=\pmat{-26.390428 & -3.0495068\\
  -3.0495068 & -30.242636}
    \prec 0_{d_2\times d_2},
    \end{align*}
    \begin{align*}
        A_{2_u}^\top \P^2 A_{2_u} - P_{2_u}
        =\pmat{-25.479355 & -3.4282391\\
  -3.4282391 & -25.646787}
  \prec 0_{d_2\times d_2}.
    \end{align*}
	We then employ Algorithm \ref{algo:impl} to generate probabilistic scheduling logics. We set \(T=1000\). It follows that \(f_{c_1} = 500\) and \(f_{c_2} = 500\). We generate \(10\) different sequences \(\gamma(0)\), \(\gamma(1),\ldots\), \(\gamma(999)\). Corresponding to each sequence, we pick \(10\) different initial conditions \(x_i^0\in[-10,+10]^{d_i}\), \(i=1,2\) and plot \(\norm{x_i(t)}^{2}\), \(i=1,2\). The resulting trajectories (up to time \(t=100\)) are illustrated in Figures \ref{fig:plant1} and \ref{fig:plant2}. Stochastic stability of each plant in the NCS under consideration follows.
    \begin{figure}
        \centering
        \includegraphics[scale = 1]{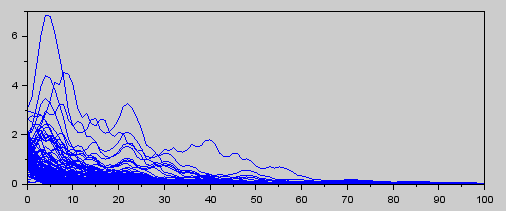}
        \caption{\(\norm{x_1(t)}^{2}\) versus \(t\)}\label{fig:plant1}
    \end{figure}
    \begin{figure}
        \centering
        \includegraphics[scale = 1]{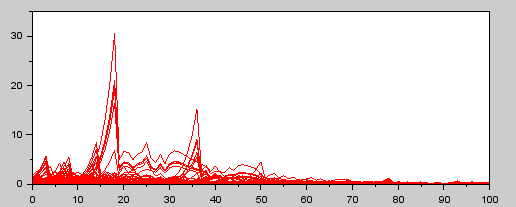}
        \caption{\(\norm{x_2(t)}^{2}\) versus \(t\)}\label{fig:plant2}
    \end{figure}
    }
    \end{experiment}
    
 Our next experiment is geared towards testing scalability of the proposed techniques.
   \begin{experiment}
   \label{ex:numex2}
   \rm{
   	We fix capacity of the shared network as \(M = 10\) and carry out the following procedure for various values of the total number of plants, \(N\), such that Assumption \ref{a:divisibility} holds:
	\begin{enumerate}[label = (\roman*), leftmargin = *]
		\item We generate unstable matrices \(A_i\in\R^{5\times 5}\) and vectors \(B_i\in\R^{5\times 1}\) with entries from the interval \([-2, 2]\) and the set \(\{0,1\}\), respectively, chosen uniformly at random and ensuring that each pair of matrices \((A_i, B_i)\), \(i = 1,2,\ldots,N\) is controllable.
		\item We compute \(v=N/M\), construct the set \(\Svec\) containing all subsets of \(\{1,2,\ldots,N\}\) with \(M\) distinct elements, choose a step size \(h=0.001\), and compute \(r\) to be biggest integer satisfying \(rh < 1\).
		\item For all distinct sets \(c_j\in\Svec\), \(j=1,2,\ldots,v\) and probabilities \(p_{c_j}\in\{h,2h,\ldots,rh\}\), \(j=1,2,\ldots,v\) satisfying \(\displaystyle{\sum_{j=1}^{v}{p}_{{c}_{j}}} = 1\), we employ Algorithm \ref{algo:controller_design} until a suitable set of state-feedback controllers, \(K_i\), \(i=1,2,\ldots,N\) are designed. We note the corresponding \(c_j\) and \(p_{c_{j}}\), \(j=1,2,\ldots,v\) and proceed to Step (iv). If no such set of controllers is found, then we report a failure.
		\item We employ Algorithm \ref{algo:impl} to generate probabilistic scheduling logics. We set \(T=1000\) and generate a sequence \(\gamma(0)\), \(\gamma(1),\ldots\), \(\gamma(999)\).
	\end{enumerate}
	
	The above set of steps was implemented by employing the LMI solver toolbox and PENBMI toolbox in MATLAB R2020a on an Intel 17-8550U, 8 GB RAM, 1 TB HDD PC with Windows 10 operating system. The time taken to conduct the experiment for various choices of \(N\) are summarized in Table \ref{tab:data_tab}. Not surprisingly, we observe that
the computation time under consideration increases as the number of plants in an NCS increases.
   }
   \end{experiment}
    \begin{table}[http]
	\centering
	{
	\begin{tabular}{|c | c | c|c|}
		\hline
		\(N\) & \(M\) & Result & Time taken (in sec)\\
		\hline
		\(100\) & \(10\) & Success & \(93\)\\
		\hline
		\(200\) & \(10\) & Success & \(1183\)\\
		\hline
		\(500\) & \(10\) & Success & \(10367\)\\
		\hline
		\(700\) & \(10\) & Success & \(33710\)\\
		\hline
		\(1000\) & \(10\) & Success & \(75726\)\\
		\hline
	\end{tabular}}
    \vspace*{0.2cm}
	\caption{Data for numerical experiment}\label{tab:data_tab}
	\end{table}
	
\section{Conclusion}
\label{s:concln}
    In this paper we presented a probabilistic algorithm to design scheduling logics for NCSs whose shared communication networks have limited capacity. We operated under the assumption that communication between plants and their controllers is not affected by any form of communication uncertainties. A next natural research direction is the design of probabilistic algorithms that construct scheduling logics for NCSs under communication uncertainties like time delays, data losses, quantization errors, etc. This matter is currently under investigation and will be reported elsewhere.






\end{document}